\documentclass[11pt,letterpaper]{article}
\usepackage{amsthm,amsmath,amssymb}
\usepackage{todonotes}
\usepackage{algorithm}
\usepackage{hyperref}
\usepackage{url}

\newcommand{\argmax}[1][0]{\underset{#1}{\arg\!\max}\;}
\newcommand{\A}{\mathcal{A}}

\theoremstyle{plain}
\newtheorem{theorem}{Theorem}
\newtheorem{lemma}{Lemma}
\newtheorem{proposition}{Proposition}

\theoremstyle{remark}
\newtheorem{example}{Example}

\newcommand{\qee}{\hfill$\lozenge$}

\usepackage{authblk}
\author[]{Wolfgang Dvo{\v r}{\' a}k}
\author[]{Monika Henzinger}
\affil[]{University of Vienna, Faculty of Computer Science, Vienna, Austria}

\title{Online Ad Assignment with an Ad Exchange}

\begin{document}

\maketitle

\begin{abstract}
Ad exchanges are becoming an increasingly popular way to sell advertisement slots on the internet.
An ad exchange is basically a spot market for ad impressions. 
A publisher who has already signed contracts reserving advertisement impressions on his pages can choose between 
assigning a new ad impression for a new page view to a contracted advertiser or to sell it at an ad exchange. This leads to an online revenue maximization problem for the publisher. 
Given a new impression to sell decide whether 
(a) to assign it to a contracted advertiser and if so to which one or 
(b) to sell it at  the ad exchange and if so at which reserve price.
We make no assumptions about the distribution of the advertiser valuations that participate in the ad exchange
and show that 
there exists a simple primal-dual based online algorithm, whose lower bound for the revenue 
converges to $R_{ADX} + R_A (1 - 1/e)$, where $R_{ADX}$ is the revenue that the optimum algorithm achieves from the ad exchange and $R_A$ is the revenue that the optimum algorithm achieves from the contracted advertisers.
\end{abstract}

% \begin{keyword}

% ad assignment \sep ad exchange \sep online assignment problem \sep revenue maximization.

% \end{keyword}

% \end{frontmatter}

\section{Introduction}
The market for display ads on the internet is worth billions of dollars and continues to rise. 
Not surprisingly, there are multiple ways of selling display advertisements. 
Traditionally, publishers signed
long-term contracts with their advertisers, fixing the number of {\em impressions}, i.e. assigned ad slots views,
as well as their price. In the last few years, however, spot markets, so called {\em Ad Exchanges}~\cite{Muthukrishnan2009}, have been developed, 
with Amazon, Ebay, 
and Yahoo (to just name a few) all offering their own ad exchange. Thus, every time a user requests to download a
page from a publisher, the publisher needs to decide
(a)  which of the ad impressions on this page should be
assigned to which contracted advertiser, and
(b) which should be sold at  the ad exchange and at which {\em reserve price}\footnote{The reserve price is the minimum required price at which an impression is sold at an ad auction. If no offer is at or above the reserve price, the impression is not sold.}.

Ad exchanges are interesting for publishers as 
(1) basically an unlimited number of ad impressions can be sold at ad exchanges, and 
(2) if the publishers have additional information about the user, 
they might sell an impression at a much higher price at the ad exchange
than they could receive from their contracted advertisers.  
As ad impressions that did not receive a bid at or above the reserve price at the ad exchange can still be
assigned to contracted advertisers, a revenue-maximizing publisher can
offer {\em every} ad impression first at an ad exchange at a ``high enough'' reserve price and then afterwards assign the still unsold impressions to contracted advertisers. 
The question for the advertiser becomes, thus, (i) what reserve price to choose, and (ii) to which advertisers to assign the unsold impressions. We model this setting as
an online problem and achieve the following two results: If the revenue achievable by the ad exchange for each ad impression is known,  
we give a constant competitive algorithm. Then we show how to
convert this algorithm into a second algorithm that works in the setting where the revenue achievable from the  ad exchange is not known. Assume that the auction 
executed at the ad exchange fulfills the following property $P$: \emph{If an ad impression is sold at the ad exchange, then the revenue achieved is independent of the reserve price chosen by the publisher.} Thus, the reserve price influences only {\em whether} the ad impression is sold, {\em not} the price that is achieved. 
For example, a first price auction with reserve prices fulfills this condition. If the auction at the ad exchange fulfills this condition, then our second algorithm is constant competitive when compared with the optimum offline algorithm.

When modeling contracted advertisers we use the {\em model with free disposal} introduced in~\cite{FeldmanKMMP09}: Each advertiser $a$ comes with a number $n_a$ and the revenue that an algorithm receives from $a$
consists of the $n_a$ {\em most valuable} ad impressions assigned to $a$. 
Additional impressions assigned to $a$ do not generate any revenue.

More formally we define the following
\emph{Online Ad Assignment Problem with Free Disposal and an Ad Exchange.} 
There is a set of contracted advertisers $A$ and an ad exchange $\alpha$. Each advertiser $a$ comes with a number
$n_a$ of ad impressions such that $a$ pays only for the $n_a$ {\em most valuable} ad impressions assigned to $a$, or for all
assigned ad impressions if fewer than $n_a$ are assigned to $a$. To simplify the notation we set $n_{\alpha} = \infty$.
Now a finite  sequence ${\cal S} = S_0, S_1, \dots $ of sets $S_l$ with 
$l = 0, 1, \dots$, of
ad impressions arrives in order. When $S_l$ arrives,  the weights $w_{i,a}$ for each $i \in S_l$ and $a \in A \cup \{\alpha\}$ are revealed and 
the online algorithm has to assign
each $i \in S_l$ {\em before} further sets $S_{l+1}, S_{l+2}$, etc.~arrive. Let $\A : I \rightarrow A \cup \{\alpha\}$ be an
{\em assignment} of impressions to advertisers. An assignment is {\em valid} if no two impressions in the same set
$S_l$ are assigned to the same advertiser $a \in A$. 
Let $I_{\A}(a)$ be the set of the $n_a$ highest weighted impressions assigned to advertiser $a$ by $\A$. 
Then the {\em revenue} $R(\A)$ of $\A$ is $\sum_{a\in A \cup \{\alpha\}} \sum_{i \in I_{\A}(a)} w_{i,a}$.
The goal of the algorithm is to produce a valid assignment $\A$ with maximum revenue $R(\A)$.
The {\em competitive ratio} of an online algorithm is the minimum over all sequences ${\cal S}$ of the ratio of
the revenue achieved by the online algorithm on ${\cal S}$ and the revenue achieved by the optimal offline algorithm on ${\cal S}$, where
the latter algorithm is given all of ${\cal S}$ {\em before} it makes the first decision.

Feldman et al.~\cite{FeldmanKMMP09} studied a special case of our problem, namely the
setting {\em without an ad exchange} and where each set $S_l$ has size one, i.e. where the impressions arrive consecutively.
For that setting they gave a primal-dual based $0.5$ competitive algorithm whose competitive ratio converges to $(1 - 1/e)$ ratio when {\em all}
the $n_a$ values go to infinity. More precisely let $n_A = \min_{a \in A} n_a$. Then their algorithm is
$1-(\frac{n_A}{n_A+1})^{n_A}$-competitive.
They also showed that this ratio is tight when considering deterministic algorithms~\cite{FeldmanKMMP09}.
 Let $R_a$ for an advertiser $a \in A \cup \{\alpha\}$ be the revenue that the optimal algorithm receives from $a$.
We extend their results in several ways. 
(1) We consider a setting with one advertiser, called ad exchange, 
that has infinite capacity\footnote{It is straightforward to extend the algorithm and its analysis to multiple ad exchanges.}.
Moreover, we allow multiple ad slots on a page, with the condition that no two can be assigned to the same
advertiser, i.e.\ for us $|S_l|$ can be larger than 1.
(2) The revenue of our algorithm depends {\em directly} on the $n_a$ value, not on $n_A$. More precisely,
if no ad exchange exists, our algorithm receives a revenue of at least $\sum_a (1-(\frac{n_a}{n_a+1})^{n_a}) R_a$.
When an ad exchange is added, our algorithm achieves a revenue of at least $R_{\alpha} + \sum_a (1-(\frac{n_a}{n_a+1})^{n_a}) R_a$.
(3) We show how to modify our algorithm for the setting where $w_{i,\alpha}$ is unknown for all $i$. 
In this setting our algorithm computes a reserve price and sends {\em every} impression first to the ad exchange. 
The reserve price is set such that if the auction executed at the ad exchange
fulfills property $P$ then the above revenue bounds continue to hold, i.e. it achieves a revenue of at least $R_{\alpha} + \sum_a (1-(\frac{n_a}{n_a+1})^{n_a}) R_a$.

\paragraph{Techniques} 
Our algorithm is a modification of the standard  primal-dual algorithms in~\cite{FeldmanKMMP09}
but it is itself not a standard primal-dual algorithms as it does not construct a feasible primal and dual solution to 
a {\em single} LP. Instead
in the analysis we use several primal and dual LPs, one for each advertiser $a$ and
use the dual solutions to upper bound $R_a$. 
However, the corresponding primal feasible solution is {\em not} directly related
to the revenue the algorithm achieves from $a$. Instead, the solution constructed by the algorithm is a feasible solution for a primal program
that is the combination of all individual LPs. This property is strong enough to give the claimed bounds.
The crucial new  ideas in our algorithms are (i)
the observation that when deciding to whom an ad slot 
is assigned the publisher should be biased towards advertiser with large $n_a$
and in particular towards the ad exchange and 
(ii) that based on the structure of the algorithm it can be easily modified to compute an reserve price for 
the auction in the ad exchange if the $w_{i,\alpha}$ values are unknown.
\paragraph{Further Related Work}
We describe prior work on the question whether the publisher should assign an impression to a contracted advertiser or an ad exchange.
In~\cite{BalseiroFMM14} a scenario is studied, where the $w_{i,a}$ follow a joint distribution and no disposal is allowed.
Gosh et al.~\cite{Ghosh2009} assume that for each impression $i$ the $w_{i,\alpha}$ values follow a
known distribution and the contracted advertisers have a quality value depending on $w_{i,\alpha}$.
They study the trade-off between the quality of the impressions assigned to the advertisers and 
revenue from the ad exchange.
The work in~\cite{Alaei2009}, like our work, does not make Bayesian assumptions but studies online algorithms in
the worst case setting. The main difference to our work 
is that there the contracted advertisers also arrive online and that there is no free disposal.

Finally, Devanur et al.~\cite{DevanurHKMY13} extend \cite{FeldmanKMMP09} to the scenario with multiple ad slots on a page and 
constraints on ads being assigned together, but they neither consider ad-exchanges nor consider the different capacities $n_a$ 
in the competitive ratio.

\paragraph{Structure of the paper}
In Section~\ref{sec:capone} we discuss why the algorithm from~\cite{FeldmanKMMP09}
is not satisfying in our setting and present a simple online algorithm 
for the 1-slot case, which we improve in  Section~\ref{sec:main} 
to achieve a revenue of at least $R_{\alpha} + \sum_a (1-(\frac{n_a}{n_a+1})^{n_a}) R_a$.
In Section~\ref{sec:multislot} we generalize this algorithm to the multi-slot setting.
Finally, in Section~\ref{sec:extensions} we show how to adapt it if the $w_{i,\alpha}$ values are unknown.
\section{A Simple 1-Slot Online Algorithm}\label{sec:capone}
In Sections \ref{sec:capone} and \ref{sec:main} we consider algorithms for the
1-slot setting, i.e., where each $S_l$ just contains a single impression $i$.
Given an instance of such an online ad assignment problem we can build an equivalent instance 
where all capacities $n_a=1$. 
Simply replace each advertiser $a$ by $n_a$ copies $a_1, \dots a_{n_a}$ with the capacities $1$ and 
for each impression $i$ set $w_{i,a_p} = w_{i,a}$  for all $1 \le p \le n_a$. 
Thus in this section we assume $n_a=1$ for each $a \in A$.
Then we formulate the offline problem as an integer linear program (ILP), 
where the variable $x_{i,a}$ is set to $1$ if $i$ is assigned to advertiser $a$ and to 0, otherwise.
\begin{align*}
 \text{\textbf{Primal:} } \max&\  \sum_{i,a \in A \cup \{\alpha\}} w_{i,a}\ x_{i,a}\\[5pt]
   \sum_{a \in A \cup \{\alpha\}} x_{i,a}&\leq 1 \quad \forall i\\
   \sum_{i} x_{i,a} &\leq 1 \quad \forall a \in A
\end{align*}
The first type of constraints ensures that each impression is assigned to at most one advertiser, while the second 
type of constraints ensures that each $a \in A$ is assigned at most one impression.
It has the following dual LP.
\begin{align*}
 \text{\textbf{Dual:} } \min&\  \sum_{i} z_i + \sum_{a \in A} \beta_a\\
   z_i+ \beta_a &\geq w_{i,a} \quad \forall i,\forall a \in A\\
   z_i &\geq w_{i,\alpha} \quad \forall i
\end{align*}

\begin{algorithm}[t]
  \caption{}
  \label{alg:1}
  \smallskip
  \begin{enumerate}
    \item Initialize $\beta_a=0$, $\beta_{\alpha}=0$
    \item When impression $i$ arrives
	\begin{enumerate}
	    \item Compute $j=\argmax[a \in A \cup \{\alpha\}]\{w_{i,a}-\beta_a\}$.
	    \item if $j=\alpha$ then set $x_{i,\alpha}=1$ and $z_i=w_{i,\alpha}$.
	    \item if $j\in A$ then set $x_{i,j}=1$, $\forall\ i'\not=i:\ x_{i',j}=0$, $z_i=w_{i,j}-\beta_j$ and $\beta_j=w_{i,j}$.
	\end{enumerate}
  \end{enumerate}\vspace{-6pt}
\end{algorithm}

For notational convenience we assume an additional variable $\beta_\alpha$  which remains $0$ for the whole algorithm.
We next consider a straight forward generalization of the online algorithm in~\cite{FeldmanKMMP09}, called Algorithm~\ref{alg:1}, 
to our setting.
This algorithm constructs a feasible integral solution for the Primal LP, corresponding to an ad assignment, 
and a feasible solution for the dual LP that is used to upper bound the revenue of the optimal assignment.

Algorithm~\ref{alg:1} constructs feasible solutions for both the Primal and the Dual:
when impression $i$ is assigned to advertiser $j$ then $x_{i,j}$ is set to 1,  $\beta_j$ is set to $w_{i,j}$, and  $z_i$ is set to $\max_{a \in A \cup \{\alpha\}}\{w_{i,a}-\beta_a\}$. 
Note that the loss in revenue of Algorithm~\ref{alg:1} compared to the optimal assignment
{\em exclusively}  comes from the impression assigned to advertisers in $A$.
\begin{proposition}\label{prop:1}
Let $\A$ be an ad assignment computed by Algorithm~\ref{alg:1},  then
$
 R(OPT) \leq R_\alpha(\A) + 2 \cdot R_A(\A)
$. 
\end{proposition}
\begin{proof}
In the following we will use that $R(\A)$ equals the value of the primal solution and 
that the value of the dual solution is an upper bound on $R(OPT)$.
We prove the claim by induction on the assigned impressions. 
Clearly the base case where no impression is assigned is fine.
Now consider an arbitrary $i$ to be assigned and notice that, by (1.), $\beta_a$ is such that $\beta_a=0$ if
no impression was assigned to $a$. 
Otherwise, by (2c), $\beta_a=w_{i',a}$ where $i'$ is the highest weighted impression assigned to $a$.
We simultaneously consider the increase of the primal and the dual solution when adding a new impression~$i$.

If Algorithm~\ref{alg:1} assigns $i$ to an $a\in A$ this is by rule $(c)$.
Let $\beta^n_a$, $\beta^o_a$ be the new and old value of $\beta_a$.
The statement $x_{i,a}=1$, $\forall\ i'\not=i:\ x_{i',a}=0$ increases the revenue $R(\A)$ by 
$w_{i,a}-\beta^o_a$, as $\beta^o_a$ is the value $w_{i',a}$ of the impression $i'$ we have to drop.
On the other side the statement $z_i=w_{i,a}-\beta^o_a$ increases the objective of the Dual by the same amount. 
Additionally the objective of the Dual is affected by  updating $\beta_a$ with $\beta^n_a=w_{i,a}$. 
This additionally increases the objective of the Dual by $\beta^n_a-\beta^o_a=w_{i,a}-\beta^o_a$.
Thus, if $i$ is assigned to an $a \in A$, the increase of the objective of the Dual is twice the increase of the objective of the primal, 
the revenue $R(\A)$.

If Algorithm~\ref{alg:1} assigns $i$ to $\alpha$ this is by rule $(b)$.
By $x_{i,\alpha}=1$ we increase the revenue by $w_{i,\alpha}$ and by setting $z_i=w_{i,\alpha}$
we increase the objective of the Dual by the same amount.
Hence, as in case (b) the $\beta_a$ are not affected,
we obtain that $R(OPT) \leq R_\alpha(\A) + 2 \cdot R_A(\A)$.
\end{proof} 

However, the  Algorithm~\ref{alg:1} does not guarantee that impressions are sent to ad exchange when the optimal algorithm does. 
Thus the optimal offline assignment might send many impressions to the ad exchange, 
while the online assignment of the above algorithm does not and thus might only be an $1/2$ approximation 
(Proposition~\ref{prop:1} guarantees that it is not worse.).
Such a situation is given in Example~\ref{example:1}.
\begin{example}\label{example:1}
Consider $A=\{a\}$ with $n_a=1$ and impressions $1 \leq i \leq n$ with  
$w_{i,\alpha}=1-\epsilon$ and $w_{i,a}=i$. 
Then the revenue  $R(\A)$ of Algorithm~\ref{alg:1} after $n$ impressions is $n$,
while the optimal assignment achieves $n+(n-1)(1-\epsilon)$, where $(n-1)(1-\epsilon)$ is achieved by the ad exchange.
For $\epsilon \rightarrow 0$ and $n \rightarrow \infty$ the ratio $R(\A)/R(OPT)$ is $1/2$ 
although half of the revenue in the optimal assignment $OPT$ comes from the ad exchange.
\qee
\end{example}
Thus the algorithm from~\cite{FeldmanKMMP09} is only $1/2$-competitive, even when an ad exchange, i.e., an advertiser with infinite
capacity, is added.

Given an ad assignment $\A$ let
$R_\alpha(\A)$ denote the revenue the assignment gets from impressions assigned to the ad exchange and let
$R_A(\A)$ denote the revenue the assignment gets from impressions assigned to contracted advertisers.
Thus we have $R(\A)=R_\alpha(\A)+R_A(\A)$. Additionally, we use $OPT$ to denote the optimal assignment.
We present next Algorithm~\ref{alg:2}, an online algorithm that 
receives as revenue at least $R_{\alpha}(OPT) + (1/2) R_a(OPT)$, which is already an improvement over Algorithm~\ref{alg:1}. 
It is based on the observation that {\em assigning an impression that should be sent to the ad exchange to an advertiser in $A$
is worse than sending an impression that should go to an  advertiser in $A$ to the ad exchange.}
Thus, the algorithm is biased towards the ad exchange.
Specifically the algorithm assigns an impression to an advertiser $a\in A$
only if it gives at least double the revenue on $a$ than on $\alpha$.

\begin{algorithm}[t]
  \caption{}
  \label{alg:2}
  \smallskip
  \begin{enumerate}
  \item Initialize $\beta_a=0$ for all $a \in A \cup \{\alpha\}$
  \item When impression $i$ arrives
	\begin{enumerate}
	    \item Compute $j=\argmax[a \in A]\{w_{i,a}-\beta_a\}$.
	    \item if $(w_{i,j}-\beta_j) > 2 \cdot w_{i,\alpha}$ then assign $i$ to $j$  and set $\beta_j=w_{i,j}$.
	    \item if $(w_{i,j}-\beta_j) \leq 2 \cdot w_{i,\alpha}$ then assign $i$ to $\alpha$.
	\end{enumerate}
  \end{enumerate}\vspace{-6pt}
\end{algorithm}

\begin{theorem}\label{thm:1}
Let $\A$ be the ad assignment computed by Algorithm~\ref{alg:2} then 
$
  R(\A) \geq R_\alpha(OPT) + 1/2 \cdot R_A(OPT)
$. 
\end{theorem}
\subsection{Proof of Theorem~\ref{thm:1}}
Let $I^A_{OPT}$, resp. $I^\alpha_{OPT}$, be the impressions assigned to $A$, resp. $\alpha$, by the optimal (offline) assignment OPT. 
\footnote{In case there are several optimal assignments we pick an arbitrary one.}
We give an LP $P_A$ for the advertisers $A$ and impressions $I^A_{OPT}$ and its dual $D_A$
such that any feasible solution for $D_A$ gives an upper bound $d_A$ for $R_A(OPT)$.
\vspace{-10pt}
\par\noindent
\begin{minipage}[t]{0.47\textwidth}
\begin{align*}
 \text{\textbf{Primal $P_A$:} } \max&\  \hspace{-15pt} \sum_{i \in I^A_{OPT},a \in A} \hspace{-15pt} w_{i,a}\ x_{i,a}\\[5pt]
   \sum_{a \in A } x_{i,a}&\leq 1 \quad \forall i\!\in\!I^A_{OPT}\\
   \sum_{i \in I^A_{OPT}} x_{i,a} &\leq 1 \quad \forall a\!\in\!A
\end{align*}
\end{minipage}
\hspace{10pt}
\begin{minipage}[t]{0.47\textwidth}
\begin{align*}
 \text{\textbf{Dual $D_A$:}}\,\min&\  \sum_{i  \in I^A_{OPT}} z_i + \sum_{a \in A} \beta_a\\[5pt]
   z_i+ \beta_a \geq w_{i,a}& \quad \forall i\!\in\!I^A_{OPT}\ \forall a\!\in\!A\\
\end{align*}
\end{minipage}

\vspace{7pt}
\noindent
Note that the summation in $P_A$ and the constraints in $D_A$ are 
only over impressions in $I^A_{OPT}$.
The objective value of the optimal solution of $D_A$,
is an upper bound for the objective of $P_A$, and thus also for $R_A(OPT)$.
However, there is no direct relationship between $R_A(\A)$ and the objective of $P_A$ for $\A$,
as $\A$ might also assign impressions from $I^\alpha_{OPT}$ to $A$.

To upper bound  $R_A(OPT)$ we construct a \emph{feasible solution for $D_A$}.
We do this in a iterative fashion, that is whenever Algorithm~\ref{alg:2} assigns an impression $i \in I^A_{OPT}$
we update the feasible solution for $D_A$ as follows:
\begin{enumerate}
 \item[(i)] For the $\beta_a$ variables we use the values currently set by the Algorithm~\ref{alg:2};
 \item[(ii)] For the variable $z_i$ we set  $z_i=w_{i,j}-\beta^o_j$, where  $\beta^o_a$ is the value of $\beta_a$ before $i$ is assigned.
\end{enumerate}
As $w_{i,j}-\beta^o_j=max_{a \in A}\{w_{i,a}-\beta_a\}$, all the constraints for $i$ are satisfied. 
Hence, doing this for all $i \in I^A_{OPT}$ gives a feasible solution for $D_A$ and its objective $d_A$ fulfills $d_A \geq R_A(OPT)$.\smallskip

To show that the claim holds after each assignment of an impression $i$ 
we investigate assigning one expression $i$ and study the effect to both the 
upper bound and the revenue we achieve.
To this end we introduce some notation:
\begin{itemize}
 \item $\Delta d_A(i)$ is the increase of the objective $d_A$ when the 
algorithm assigns impression $i$, i.e., 
the change in $d_A$ caused by the change in the $\beta$-values and the assignment of the new $z_i$ value (if $i \in I^A_{OPT}$).
 \item $\Delta d_\alpha(i)=w_{i,\alpha}$ if $i \in I^\alpha_{OPT}$ and $\Delta d_\alpha(i)=0$ otherwise.
 \item $\Delta R(\A,i)$ is the increase of the revenue when assigning $i$.
 \item $\beta^n_a$, resp.\ $\beta^o_a$, to denote the value of $\beta_a$ after, resp.\ before $i$ is assigned
\end{itemize}
Note that by the definitions  
(a) $\sum_{i \in I} \Delta d_A(i)=d_A$, 
(b) $\sum_{i \in I} \Delta d_\alpha(i)= R_\alpha(OPT)$ and 
(c) $\sum_{i \in I} \Delta R(\A,i)= R(\A)$. 
We will also exploit the fact that $\beta_a$ is such that $\beta_a=0$ if no impression was assigned to $a$ and otherwise 
$\beta_a=w_{i',a}$, where $i'$ is the impression currently assigned to $a$.

Next, to relate the increase of the upper bound with the gain of revenue, we distinguish whether one assigns an impression $i \in I^\alpha_{OPT}$ or an impression $i \in I^A_{OPT}$  

\begin{lemma}\label{theorem1:lem1}
 $\Delta R(\A,i) \geq \Delta d_\alpha(i) + 1/2 \cdot \Delta d_A(i)$, for $i \in I^A_{OPT}$.
\end{lemma}
\begin{proof}
 As $i \in I^A_{OPT}$, by definition, we have $\Delta d_\alpha(i)=0$.  
Thus, we actually have to show that $\Delta R(\A,i) \geq 1/2 \cdot \Delta d_A(i)$.
\begin{enumerate}
       \item If  Algorithm~\ref{alg:2} assigns $i$  to an $j\in A$ recall that we 
	     set $z_i=w_{i,j}-\beta^o_j$ and the algorithm sets $\beta^n_j=w_{i,j}$.
	     Thus  $\Delta d_A(i)=2\cdot(w_{i,j}-\beta^o_j)$ and
	     $\Delta R(\A,i)$ is given by $w_{i,j}$ minus the value of the impression we have to drop (if any), given by $\beta^0_a$.
	     As this values is stored in $\beta^o_j$ we get $\Delta R(\A,i)=w_{i,j}-\beta^o_j$ and thus $\Delta R(\A,i) \geq 1/2 \cdot \Delta d_A(i)$.
       \item If Algorithm~\ref{alg:2} assigns $i$ to $\alpha$ (although $OPT$ does not), 
	     we know from Step 2c that $(w_{i,j}-\beta_j) \leq 2 w_{i,\alpha}$,
	     where $j =\argmax[a \in A]\{w_{i,a}-\beta_a\}$. 
	     As we set $z_i=w_{i,j}-\beta^o_j$ and the algorithm keeps all $\beta_a$ unchanged
	     we get $\Delta d_A(i)= w_{i,j}-\beta^o_j$ and as we assign $i$ to $\alpha$ we have $\Delta R(\A,i)=w_{i,\alpha}$.
	     Thus $\Delta R(\A,i) = w_{i,\alpha} \geq 1/2 \cdot (w_{i,j}-\beta_j) = 1/2 \cdot \Delta d_A(i)$.
\end{enumerate} 
Hence, 
$\Delta R(\A,i) \geq 1/2 \cdot \Delta d_A(i)= \Delta d_\alpha(i) + 1/2 \cdot \Delta d_A(i)$. 
\end{proof}

\begin{lemma}\label{theorem1:lem2}
 $\Delta R(\A,i) \geq \Delta d_\alpha(i) + 1/2 \cdot \Delta d_A(i)$, for $i \in I^\alpha_{OPT}$.
\end{lemma}

\begin{proof}
As $i \in I^\alpha_{OPT}$, by definition, $\Delta d_\alpha(i)=w_{i,\alpha}$. 
Recall that no $z$-value is affected in this case. 
We have to show that $\Delta R(\A,i) \geq w_{i,\alpha} + 1/2 \cdot \Delta d_A(i)$.
\begin{enumerate}
 \item If Algorithm~\ref{alg:2} assigns $i$ to the ad exchange then the $\beta_a$ are not changed.
    Thus $\Delta d_A(i)=0$ and
    $\Delta R(\A,i)$ is simply $w_{i,\alpha}$. Hence, $\Delta R(\A,i) = w_{i,\alpha} \geq w_{i,\alpha} + 1/2 \cdot \Delta d_A(i)$.
 \item  If Algorithm~\ref{alg:2} assigns $i$ to an $a \in A$ we have $(w_{i,a}-\beta^o_a) > 2 w_{i,\alpha}$
 and the algorithm sets $\beta^n_a=w_{i,a}$.
 Thus $\Delta d_A(i)=w_{i,a}-\beta^o_a$.
 Furthermore, $\Delta R(\A,i)$ is given by $w_{i,a}$ minus the value of the impression we have to drop (if any), given by $\beta^0_a$.
 Thus $\Delta R(\A,i)= (w_{i,a}-\beta^o_a) =1/2 \cdot (w_{i,a}-\beta^o_a)+1/2 \cdot(w_{i,a}-\beta^o_a) \geq w_{i,\alpha} + 1/2 \cdot \Delta d_A(i)$.  
\end{enumerate}
Hence,
$\Delta R(\A,i) \geq w_{i,\alpha} + 1/2 \cdot \Delta d_A(i)= \Delta d_\alpha(i) + 1/2 \cdot \Delta d_A(i)$.
\end{proof}

When combining Lemma~\ref{theorem1:lem1} and Lemma~\ref{theorem1:lem2} we obtain 
that for each $i\in I$:
$$
  \Delta R(\A,i) \geq \Delta d_\alpha(i) + 1/2 \cdot \Delta d_A(i).
$$

Finally, when summing over all impression $i \in I$ and using the above inequality we obtain the claim.

\begin{align*}
 R(\A)\!=\!\sum_{i \in I} \Delta R(\A,i) \geq  \sum_{i \in I} \left( \Delta d_\alpha(i)\!+\!\frac{\Delta d_A(i)}{2}\right)\, 
      \geq R_\alpha(OPT)\!+ \frac{R_A(OPT)}{2}.
\end{align*}
This completes the proof of Theorem~\ref{thm:1}.

\section{An Online 1-Slot Algorithm Exploiting High Capacities}\label{sec:main}
In this section we generalize the result from Section~\ref{sec:capone} to the setting where each advertiser $a\in A$ 
has an individual limit $n_a$ for the number of ad impressions he is willing to pay for
and we present Algorithm~\ref{alg:3} that achieves an improvement in revenue for advertisers $a$ with large $n_a$.

\vspace{7pt}
\noindent
In Algorithm~\ref{alg:3} we consider variables $\beta_a$ which, for $a \in A$, are always set s.t.
\begin{equation}\label{equ:beta}
 \beta_a = \frac{1}{n_a (e_{n_a}-1)} \sum_{j=1}^{n_a} w_j \left( 1+\frac{1}{n_a} \right)^{j-1}
\end{equation}
where the $w_j$'s are the weights of the impressions assigned to $a$ in non-increasing order
and $e_{n_a}=(1+1/{n_a})^{n_a}$.
That is, $\beta_a$ stores a weighted mean of the $n_a$ most valuable impressions assigned to $a$.
Again we keep $\beta_\alpha\!=\!0$ in the whole algorithm. 
Next we consider how assigning a new impression to $a$ affects $\beta_a$.
\begin{lemma}\label{lem:beta}
 	Consider a new impression $i$ being assigned to advertiser $a$. 
 	Let $\beta_a^o$, resp.\ $\beta_a^n$ denote the value of $\beta_a$ before, resp.\ after $i$ was assigned
 	and $v$ the value of the impression dropped from $\beta_a$ ($0$ if no impression is dropped), then
	\[
	    \beta_a^n - \beta_a^o\leq \frac{\beta_a^o}{n_a} - \frac{v \cdot e_{n_a}}{n_a (e_{n_a}-1)} + \frac{w_{i,a}}{n_a (e_{n_a}-1)}.
	\]
\end{lemma}

Lemma~\ref{lem:beta} was already shown in~\cite{FeldmanKMMP09} but to keep the paper self-explanatory we provide a proof.
\begin{proof}
 Assume that $w_{i,a}$ is the impression with the k-th highest value.
 \begin{align*}
  \beta_a^n &=\!\frac{1}{n_a (e_{n_a}\hspace{-3pt}\!-\!1)} 
	  \left[ \sum_{j=1}^{k-1} w_j \left(\!1\!+\!\frac{1}{n_a} \right)^{j\!-\!1} \hspace{-12pt} +
		 w_{i,a}\left(\!1\!+\!\frac{1}{n_a}\right)^{k\!-\!1} \hspace{-5pt} +
		 \sum_{j=k}^{n_a-1} w_j \left(\!1\!+\!\frac{1}{n_a} \right)^{j}	  
	  \right]\nonumber\\
	  &\leq \!\frac{1}{n_a (e_{n_a}\hspace{-3pt}\!-\!1)} 
	  \left[ w_{i,a} + \sum_{j=1}^{n_a-1} w_j \left(\!1\!+\!\frac{1}{n_a} \right)^{j}  
	  \right]
 \end{align*}
 To obtain the inequality we exploited that $w_j\!>\!w_{i,a}$ for $j\!<\!k$ and that
 $\left( 1 + {1}/{n_a} \right)^{j-1} < \left( 1+{1}/{n_a} \right)^{j'-1}$ for $j<j'$. 
 We proceed with standard transformations.
 \begin{align*}
\beta_a^n &\leq \!\frac{1}{n_a (e_{n_a}\hspace{-3pt}\!-\!1)} 
	  \left[ w_{i,a} + \sum_{j=1}^{n_a-1} w_j \left(\!1\!+\!\frac{1}{n_a} \right)^{j}  
	  \right]\\
	  &= \!\frac{w_{i,a}}{n_a (e_{n_a}\hspace{-3pt}\!-\!1)}+ \!\frac{1}{n_a (e_{n_a}\hspace{-3pt}\!-\!1)} \sum_{j=1}^{n_a-1} w_j \left(\!1\!+\!\frac{1}{n_a} \right)^{j}\\
	  &=\!\frac{w_{i,a}}{n_a (e_{n_a}\hspace{-3pt}\!-\!1)}+ \left(1+\frac{1}{n_a}\right)\beta_a^o - \frac{v \cdot e_{n_a}}{n_a (e_{n_a}-1)}
 \end{align*}
 From the last statement we obtain 
 $\beta_a^n - \beta_a^o\leq \frac{\beta_a^o}{n_a} - \frac{v \cdot e_{n_a}}{n_a (e_{n_a}-1)} + \frac{w_{i,a}}{n_a (e_{n_a}-1)}$.
\end{proof}
\begin{algorithm}[t]
  \caption{}
  \label{alg:3}
  \smallskip
  \begin{enumerate}
  \item Initialize $\beta_a=0$ for all $a \in A \cup \{\alpha\}$
  \item When impression $i$ arrives
	\begin{enumerate}
	    \item Compute $x=\argmax[a \in A \cup \{\alpha\}]\{c_a\cdot(w_{i,a}-\beta_a)\}$		    
	    \item assign $i$ to $x$ and update $\beta_x$ according to (\ref{equ:beta})
	\end{enumerate}
  \end{enumerate}
  where weights $c_a$ are defined as
$
 c_a=\begin{cases}
      1-\frac{1}{e_{n_a}} & a\in A\\
      1 & a=\alpha
     \end{cases}
$
\end{algorithm}

Notice that in Algorithm~\ref{alg:3} for each  $a \in A$ we have that $1/2 \leq c_a < 1 - 1/e$, i.e.\
for $n_a=1$ we have $c_a=1/2$ and $c_a$ grows with $n_a$ and converges to $1 - 1/e$.
The idea is to bias the algorithm towards advertisers with larger $n_a$ and in particular towards the 
ad-exchange.

We use $R_a(\A)$ for $a \in A \cup \{\alpha\}$ to denote 
the revenue the assignment $\A$ gets from advertiser $a$.
Thus, $R(\A) = \sum_{a \in A \cup \{\alpha\}}  R_a(\A)$.

\begin{theorem}\label{thm:main}
Let $\A$ be the assignment computed by Algorithm~\ref{alg:3}
then
$
  R(\A) \geq \sum_{a \in A \cup \{\alpha\}} c_a \cdot R_a(OPT)
$. 
\end{theorem}
Theorem~\ref{thm:main} will be a direct consequence of Theorem~\ref{thm:multislot}, which we will prove in the next section. 

Finally let us briefly discuss whether the constants $c_a$ are chosen optimally.
From a result in~\cite{KalyanasundaramP96} on online algorithms for $b$-matchings it follows immediately 
that the constants $c_a$ in Theorem~\ref{thm:main} are optimal for deterministic algorithms.
Moreover, in \cite{MehtaSVV07} it is shown that even randomized algorithms cannot achieve a better competitive ratio than $(1-1/e)$~\footnote{In \cite{MehtaSVV07} the authors study the Adwords problem but in \cite{FeldmanKMMP09} it is argued that
the given example can be also be interpreted as Online Ad Assignment problem.}.
So for large values of $n_a$ even randomized algorithms cannot improve over Algorithm~\ref{alg:3}.

\section{A Multi-Slot Online Algorithm}\label{sec:multislot}
\newcommand{\ax}[0]{\mathbf{a}}
\newcommand{\bx}[0]{\mathbf{b}}
In practice publishers often have several ad slots at a single page and want to avoid to show multiple ads from the same
advertiser on the same page to avoid annoying their users. 
This can be modeled as follows: A sequence ${\cal S} = S_0, S_1, \dots$ of
{\em sets} of impressions arrive in an online manner. Each set $S$ has be assigned 
(a) before any future sets have arrived, and 
(b) such that
non two impressions in $S$ are assigned to the same advertiser in $A$. 
Note that we allow multiple impressions from $S$ to be assigned to the ad exchange
as we expect the ad exchange to return different advertisers for them. 
Let the set of all impressions $I = \sum_{S \in {\cal S}} S$.
With Algorithm~\ref{alg:4} we present an online algorithm for this setting with the same competitive ratio as 
Algorithm~\ref{alg:3}. Note, however, that, unlike Algorithm~\ref{alg:3}, it is compared to the optimal offline solution that respects the above restriction. 
More formally,
we call a function $\ax:S \rightarrow A \cup \{\alpha\}$ assigning impressions $S$ to advertisers \emph{valid} 
if there are no $i,i'\in S$, $i\not=i'$, $a\in A$ such that $\ax(i)=\ax(i')=a$.
Our Algorithm~\ref{alg:4} generates a valid assignment and is compared to the revenue of the {\em valid} assignment generated by the optimal offline algorithm. 
Notice that the computation of ${\arg\!\max}$ in Algorithm~\ref{alg:4} is a weighted bipartite matching problem and 
thus can be computed efficiently.

\begin{algorithm}[t]
  \caption{}
  \label{alg:4}
  \smallskip
  \begin{enumerate}
  \item Initialize $\beta_a=0$ for all $a \in A \cup \{\alpha\}$
  \item When impressions $S=\{i_1,\dots,i_l\}$ arrive
	\begin{enumerate}
	    \item Compute $\displaystyle \bx=\argmax[\text{valid}\ \ax]\left\{ \sum _{i\in S} c_{\ax(i)} \cdot(w_{i,\ax(i)}-\beta_{\ax(i)})\right\}$		    
	    \item assign each $i$ to $\bx(i)$ and, if $\bx(i)\in A$, update $\beta_{\bx(i)}$ according to (\ref{equ:beta}).
	\end{enumerate}
  \end{enumerate}
  where weights $c_a$ are defined as
$
 c_a=\begin{cases}
      1-\frac{1}{e_{n_a}} & a\in A\\
      1 & a=\alpha
     \end{cases}
$
\end{algorithm}
Recall that $R_a(OPT)$ for $a \in A \cup \{\alpha \}$ is  the revenue that an optimal assignment generates from advertiser $a$.
We show the following performance bound.
\begin{theorem}\label{thm:multislot}
Let $\A$ be the assignment computed by Algorithm~\ref{alg:4}
and OPT the optimal multi-slot ad assignment,
then
$
  R(\A) \geq \sum_{a \in A \cup \{\alpha\}} c_a \cdot R_a(OPT)
$. 
\end{theorem}

The proof of Theorem~\ref{thm:multislot} generalizes ideas from the proof of Theorem~\ref{thm:1}. 
One of the main difference is that we now have to deal with several sets $I^a_{OPT}$ instead of just one set $I^A_{OPT}$ and thus also with several LPs.

\subsection{Proof of Theorem~\ref{thm:multislot}}
First we give a linear program $P_a$ and its dual $D_a$ for each $a \in A$ such that the final objective  value  of 
any feasible solution of $D_a$ is an upper bound of $R_a(OPT)$.
Note, that 
there is no direct relationship between the final objective values of the  $P_a$'s  and the revenue of the algorithm. 
However, we are able to construct a feasible solution for each $D_a$ with objective value $d_a$ such that 
the revenue $R(\A)$ of the algorithm is at least $\sum_{a \in A \cup {\alpha}}c_{a} \cdot d_{a}$. 
Together with the observation that $d_a \geq R_a(OPT)$ and a bound $d_\alpha$ on $R_\alpha(OPT)$ this proves the theorem.

Let $I^a_{OPT}$ be the impressions assigned to $a \in A \cup \{\alpha\}$ by the optimal (offline) assignment OPT.
We consider the following LPs for each $a \in A$. 
\vspace{-10pt}
\par\noindent
\begin{minipage}[t]{0.47\textwidth}
\begin{align*}
 \text{\textbf{Primal $P_a$:} } \max& \  \sum_{i \in I^a_{OPT}} w_{i,a}\ x_{i,a}\\[5pt]
   x_{i,a} & \leq 1 \quad \forall i \in I^a_{OPT}\\
   \sum_{i \in I^a_{OPT}} x_{i,a} &\leq n_a
\end{align*}
\end{minipage}
\hspace{10pt}
\begin{minipage}[t]{0.47\textwidth}
\begin{align*}
 \text{\textbf{Dual $D_a$:} } \min&\  \sum_{i  \in I^a_{OPT}} z_i + n_a \beta_a\\[5pt]
   z_i+ \beta_a &\geq w_{i,a} \quad \forall i  \in I^a_{OPT}
\end{align*}
\end{minipage}

\vspace{7pt}
\noindent
Note that the summation in the primal and the constraints in the Dual are only over the
impressions {\em in $I^a_{OPT}$}, i.e., the impressions assigned to $a$ by OPT. 
The objective value of the optimal solution for $D_a$
is an upper bound for the objective of $P_a$, and thus also for $R_a(OPT)$. 
This implies that any feasible solution of $D_a$, also the one we construct next, gives an upper bound for $R_a(OPT)$.
As there might be impressions assigned to $a$ by the algorithm that do {\em not} belong to $I^a_{OPT}$, 
the objective value of $P_a$ is, however, not necessarily related to $R_a(\A)$.

Next, we give a \emph{feasible solution for $D_a$} for all $a \in A$, as follows. 
\begin{enumerate}
 \item[(i)] For the $\beta_a$ variables we use the values currently set by the Algorithm~\ref{alg:4};
 \item[(ii)] Let $\ax$ be the assignment of the impressions in $S$ by the optimal solution.
	     For each $i \in I$, we set $z_i=w_{i,\ax(i)}-\beta_{\ax(i)}$ exactly when the algorithm assigns $i$. 
\end{enumerate}
Note that this results in a feasible dual solution for {\em all} $a$ as each $i$ belongs to exactly one set $I^{\ax(i)}_{OPT}$ and $z_i$  is chosen
exactly so as to make
the solution of $D_{\ax(i)}$ feasible, together with the current $\beta_{\ax(i)}$ values. As $\beta_{\ax(i)}$ only increases in the
course of the algorithm the solution remains feasible at the end of the algorithm.
Let $d_a$ be the value of this feasible solution for $D_a$ for some $a\in A$ then we have $d_a \ge R_a(OPT)$.

To show that the claim holds after each assignment of a set of impressions $S$ 
we investigate assigning one such set and study the effect to both the  upper bound and the revenue we achieve.
To this end we introduce some notation:
\begin{itemize}
 \item $\Delta d_a(S)$ be the increase of the objective value $d_a$ when the algorithms assigns $S$, 
	i.e., the change in $d_a$ caused by the change in the $\beta_a$-values 
	\emph{and} the assignment of the $z_i$-values for all $i \in S \cap I^A_{OPT}$. 
 \item $\Delta d_\alpha(S)=\sum_{i \in S \cap I^\alpha_{OPT}} w_{i,\alpha}$.
 \item $\Delta R(\A,S)$ is the increase of revenue when $S$ is assigned.
\end{itemize}
Note that by the definitions  
(a) $\sum_{S \in {\cal S}} \Delta d_a(S) = d_a$, 
(b) $\sum_{S \in {\cal S}} \Delta d_\alpha(S) = d_\alpha$ and 
(c) $\sum_{S \in {\cal S}} \Delta R(\A,S)= R(\A)$. 

\begin{lemma}\label{lem:multislot}
$\displaystyle \Delta R(\A,S) \ge  \sum_{a \in A \cup {\alpha}}c_{a} \cdot \Delta d_{a}(S)$, for all  $S \in {\cal S}$.
\end{lemma}
\begin{proof}
To simplify the notation let $\Delta d(S) = \sum_{a \in A \cup {\alpha}}c_{a} \cdot \Delta d_{a}(S)$.\smallskip

\emph{First consider $\Delta R(\A,S)$}: 
 For $a\in A$ let $v_a$ be the value of the $n_a$-th valuable impression assigned to $a$ 
 (the impression we would ``drop'' by assigning a new one),
 and let $v_\alpha=0$. 
 If $i$ is assigned to $\alpha$ then the gain in revenue is $w_{i,\bx(i)}$ which equals $w_{i,\bx(i)}-v_{\bx(i)}$.
 If $i$ is assigned to $a \in A$ then the gain in revenue is the difference between the revenue of the new impression and 
 the impression we have to drop, i.e., again $w_{i,\bx(i)}-v_{\bx(i)}$. Thus for $S$ altogether it holds
      \begin{equation*}
	  \Delta R(\A,S)= \sum_{i \in S} (w_{i,\bx(i)}-v_{\bx(i)})
      \end{equation*}\smallskip
 
 \emph{Now consider $\Delta d(S)$:}
 Recall that $\ax$ is the assignment of the optimal solution for the impressions $S$
 and let $\bx$ be the assignment from Algorithm~\ref{alg:4}.
 For all $a \in A$ let $\beta_a^o$, $\beta_a^n$ denote the value of $\beta_a$ right {\em before}, resp.~right {\em after} this assignment.
Recall that for $a = \alpha$, it holds that $\beta_a =0$ throughout the algorithm.
 Now note that
       \begin{equation*}
	  \Delta d(S) = \sum_{i \in S}\left( c_{\ax(i)} \cdot (w_{i,\ax(i)}-\beta^o_{\ax(i)}) + c_{\bx(i)} \cdot n_{\bx(i)} \cdot (\beta^n_{\bx(i)}-\beta^o_{\bx(i)})\right),
      \end{equation*}
 where the first term comes from the new variables $z_i$ which we set to $(w_{i,\ax(i)}-\beta^o_{\ax(i)})$, and 
 the second term comes from the updates of $\beta_a$.
 By the choice of $\bx$ in the algorithm we get
       \begin{align*}
	  \Delta d(S) &\leq \sum_{i \in S} \left( c_{\bx(i)} \cdot (w_{i,\bx(i)}-\beta^o_{\bx(i)}) + c_{\bx(i)} \cdot n_{\bx(i)} \cdot (\beta^n_{\bx(i)}-\beta^o_{\bx(i)})\right)\\
		&=\sum_{i \in S}  c_{\bx(i)} \cdot  \left((w_{i,\bx(i)}-\beta^o_{\bx(i)}) + n_{\bx(i)} \cdot (\beta^n_{\bx(i)}-\beta^o_{\bx(i)})\right).
      \end{align*} 
 Next we bound the contribution of each $i \in S$  separately 
 by analyzing two cases:
 \begin{itemize}
    \item If $\bx(i)=\alpha$ then we know that $\beta^o_{\bx(i)}=\beta^n_{\bx(i)}=v_{\bx(i)}=0$ and $c_{\bx(i)}=1$ . Thus
             \begin{equation*}
		  c_{\bx(i)} \cdot  \left( (w_{i,\bx(i)}-\beta^o_{\bx(i)}) + c_{\bx(i)} \cdot n_{\bx(i)} \cdot (\beta^n_{\bx(i)}-\beta^o_{\bx(i)})\right)
		  = (w_{i,\bx(i)}-v_{\bx(i)}).
	     \end{equation*} 
    \item If $\bx(i)\in A$ then we can apply Lemma~\ref{lem:beta} to bound $(\beta^n_{\bx(i)}-\beta^o_{\bx(i)})$ as follows
             \begin{align*}
		  c_{\bx(i)} \cdot \left( (w_{i,\bx(i)}-\beta^o_{\bx(i)}) + n_{\bx(i)} \cdot (\beta^n_{\bx(i)}-\beta^o_{\bx(i)})\right) &\leq\\
		  c_{\bx(i)} \cdot \left( (w_{i,\bx(i)}-\beta^o_{\bx(i)}) + \beta^o_{\bx(i)} - \frac{v_{\bx(i)} \cdot e_{n_{\bx(i)}}}{e_{n_{\bx(i)}}-1}+\frac{w_{i,\bx(i)}}{e_{n_{\bx(i)}}-1}\right) &=\\
		  c_{\bx(i)} \cdot \left( \frac{w_{i,\bx(i)}\cdot e_{n_{\bx(i)}}}{e_{n_{\bx(i)}}-1} - \frac{v_{\bx(i)} \cdot e_{n_{\bx(i)}}}{e_{n_{\bx(i)}}-1}\right)=(w_{i,\bx(i)}-v_{\bx(i)})
	     \end{align*} 
 \end{itemize} 
 In the last step we used that by definition $c_a=1-{1}/{e_{n_a}}$ for $a \in A$.
 By the above we obtain
      \begin{equation*}
	  \Delta d(S) \leq \sum_{i \in S}(w_{i,\bx(i)}-v_{\bx(i)}) =\Delta R(\A, S).\qedhere
      \end{equation*}
 \end{proof}
Now consider that the set of impression is given by a series $(S_j)_{0\leq j \leq n}$ of pairwise disjoint sets of impressions 
that show up simultaneously.
Exploiting Lemma~\ref{lem:multislot} we get:
  \begin{align*}
    R(\A) =& \sum_{j=0}^n\Delta R(\A, S_j) \geq \sum_{j=0}^n\Delta d(S_j)=\\
    & \sum_{j=0}^n \sum_{a \in A \cup \{\alpha\}} c_{a} \cdot \Delta d_{a}(S_j) \geq \sum_{a \in A \cup \{\alpha\}} c_a \cdot R_a(OPT)
  \end{align*}
This completes the proof of Theorem~\ref{thm:multislot}.

\section{An Algorithm for Computing Reserve Prices}
\label{sec:extensions}

In our model we assumed the publisher knows exactly how much revenue he can get from the ad exchange,
i.e., the $w_{i,\alpha}$ values are given for all $i \in I$.
The critical reader may interpose that this is not the fact in the real world 
or in the ad exchange model proposed in~\cite{Muthukrishnan2009}. 
Instead whenever sending an impression to the ad exchange an auction 
is run.
However, the publisher can set a reserve price and if all the bids are below the reserve price then 
he can still assign it to one of the contracted advertisers.

One nice property of Algorithms \ref{alg:2} \& \ref{alg:3} is that
they allow to compute the minimal price we have to extract from the ad exchange such that 
it is better to assign an impression to the ad exchange
than to a contracted advertiser. 
This price is given by $\max_{a\in A}\left\{c_a\cdot (w_{i,a}-\beta_a)\right\}$.
It follows that this price is also a natural choice for the reserve price.
Assume the auction executed at the ad exchange fulfills the following \emph{property
(P): If an ad impression is sold at the ad exchange, then the revenue achieved is independent of the reserve price chosen by the publisher.} 
Thus, the reserve price influences only {whether} the ad impression is sold, {not} the price that is achieved. 
Then Theorem~\ref{thm:multislot} applies, i.e., the revenue of the algorithm is at least
$\sum_{a \in A \cup \{\alpha\}} c_a \cdot R_a(OPT)$, even though the algorithm is not given the $w_{i,\alpha}$
values and it is compared to an optimal algorithm that does.
The reason is that the algorithm  makes exactly the same decisions and receives exactly the same revenue as
Algorithm~\ref{thm:multislot} that is given the $w_{i,\alpha}$ values.
\begin{theorem}\label{thm:reserveprice}
Let $\A$ be the assignment computed by the Algorithm described above, i.e., without knowledge of the $w_{i,\alpha}$ values.
If the auction at the ad exchange fulfills property P, then
$
  R(\A) \geq \sum_{a \in A \cup \{\alpha\}} c_a \cdot R_a(OPT)
$. 
\end{theorem}
However, this technique does not work for Algorithm \ref{alg:4}.
When a set of impressions $S$ is assigned the right reserve price for an impression $i \in S$ depends
on the assignment of the other impressions in $S$. In particular the optimal reserve price for one impression 
might depend on the outcome of an auction for another item in the same set. 
We leave integrating reserve prices in the multi-slot setting as an open question for future research.

\section*{Acknowledgments}
The research leading to these results has received funding from the European Research Council 
under the European Union's Seventh Framework Programme (FP/2007-2013) / ERC Grant
Agreement no.\ 340506 and from
the Vienna Science and Technology Fund (WWTF) through project ICT10-002.

The authors are grateful to Claire Kenyon and Moses Charikar for useful discussions on formulating the model.

A preliminary version of this paper has appeared in the proceedings of the 12th International Workshop on Approximation and Online Algorithms~\cite{DvorakH14}.

\bibliographystyle{elsarticle-num}

\end{document}